%% file: main.tex
\documentclass[conference]{IEEEtran}
\IEEEoverridecommandlockouts

\usepackage{cite}
\usepackage{amsmath,amssymb,amsfonts}
\usepackage{amsthm}
\usepackage{algorithmic}
\usepackage{graphicx}
\usepackage{textcomp}
\usepackage{xcolor}
\def\BibTeX{{\rm B\kern-.05em{\sc i\kern-.025em b}\kern-.08em
    T\kern-.1667em\lower.7ex\hbox{E}\kern-.125emX}}
\usepackage{acronym}
\usepackage{times}
\usepackage{epsfig,verbatim}

\usepackage[]{mdframed}
\usepackage{tikz}
\usepackage{mathtools}
\usepackage{scalerel}
\usepackage{bbm}
\usepackage{caption}
\usepackage{subcaption}
\usepackage{array}
\usepackage{pgfplots}
\pgfplotsset{compat=newest}
\usepackage{cuted}
\usepackage{stmaryrd}
\usepackage{layout}

\allowdisplaybreaks

\newcommand{\indep}{\perp \!\!\! \perp}
\newcommand{\brackets}[1]{\left(#1\right)}
\newcommand{\sbrackets}[1]{\left[#1\right]}
\newtheorem{theorem}{Theorem}

\newtheorem{lemma}{Lemma}
\newtheorem{proposition}{Proposition}
\newtheorem{definition}{Definition}
\newtheorem{remark}{Remark}
\newtheorem*{assumption}{Assumption A}

\acrodef{mmse}[MMSE]{{minimum mean-squared error}}
\acrodef{dm}[DM]{{decision maker}}
\acrodef{lqg}[LQG]{Linear-Quadratic-Gaussian}
\acrodef{pdf}[p.d.f.]{probability density function}
\acrodef{dpc}[DPC]{Dirty Paper Coding}
\acrodef{dmc}[DMC]{discrete memoryless channel}

\newcommand{\abs}[1]{\left\lvert#1\right\rvert}

\begin{document}

\title{Empirical Coordination over Markov Channel \\with Independent Source
\thanks{This work is supported by Swedish Research Council (VR) under grant 2020-03884. The work of M. Le Treust is supported by the French National Agency for Research (ANR) via the project n°ANR-22-PEFT-0010 of the France 2030 program PEPR Futur Networks.}
}

\author{%
  \IEEEauthorblockN{Mengyuan Zhao\IEEEauthorrefmark{1}, Maël Le Treust\IEEEauthorrefmark{2}, Tobias J. Oechtering\IEEEauthorrefmark{1}}
  \IEEEauthorblockA{\IEEEauthorrefmark{1}KTH Royal Institute of Technology, 100 44
Stockholm, Sweden, Email: \{mzhao, oech\}@kth.se}

   \IEEEauthorblockA{\IEEEauthorrefmark{2}CNRS, University of Rennes, Inria, IRISA UMR 6074, F-35000 Rennes, France, Email: mael.le-treust@cnrs.fr}
}

\maketitle

\begin{abstract}
We study joint source–channel coding over Markov channels through the empirical coordination framework. More specifically, we aim at determining the empirical distributions of source and channel symbols that can be induced by a coding scheme. We consider strictly causal encoders that generate channel inputs, without access to the past channel states, henceforth driving the Markov state evolution. Our main result is the single-letter inner and outer bounds of the set of achievable joint distributions, coordinating all the symbols in the network. To establish the inner bound, we introduce a new notion of typicality, the input-driven Markov typicality, and develop its fundamental properties. Contrary to the classical block-Markov coding schemes that rely on the blockwise independence for discrete memoryless channels, our analysis directly exploits the Markov channel structure and improves beyond the independence-based arguments. 
\end{abstract}

\begin{IEEEkeywords}
coordination coding, joint source-channel coding, Markov channel, input-driven Markov typicality.
\end{IEEEkeywords}

\section{Introduction}


Markov chains, as canonical models of stochastic recursion, are often represented as being driven by exogenous independent randomness \cite{diaconis1999iterated}. This perspective casts the evolution as a state update rule acted on by fresh external randomness at each step, unifying formulations across random dynamical systems and stochastic algorithms. When combined with structural conditions such as monotonicity or contractivity, it leads to tractable analyses of stationarity and convergence \cite{hermer2023rates,gupta2024probabilistic}. Moreover, from a control perspective, such exogenous inputs correspond to the classical open-loop operation \cite{aastrom2012introduction,bertsekas2012dynamic,tanaka2012dynamic}.



In information theory, Markov channels, sometimes known as finite state channels (FSC), have been studied extensively through the lens of channel capacity and control mechanisms \cite{blackwell1958proof,gallager1968information,gray1987ergodicity,verdu1994general,goldsmith2002capacity}. When feedback is available, the encoder can adapt its inputs based on past channel outputs, leading to a rich body of results on FSCs with feedback \cite{massey1990causality, permuter2009finite,tatikonda2008capacity,grigorescu2024capacity}. The seminal work \cite{chen2005capacity} characterized the capacity via directed information, with the structural assumption of the capacity-achieving distribution later established in \cite{kourtellaris2017information}. Specific cases, such as the binary symmetric Markov channel was studied in \cite{permuter2014capacity}, and further generalized in \cite{stavrou2017sequential}.
These capacity characterizations, however, are generally intricate and involve multi-letter expressions. Recently, under the unichain and ergodicity assumptions, a single-letter expression was obtained in \cite{wu2025actions}.

In this work, we study the Markov evolution structure from a joint source–channel coding (JSCC) perspective, as in \cite{shannon1948mathematical,gunduz2024joint,kourtellaris2015nonanticipative}. 
Specifically,  we consider that the source symbol is generated independently at each time, and the channel input is generated by strictly causal encoders based solely on the past source sequence, and does \textit{not} adapt to the latent channel state history, i.e., no channel feedback. The channel input, together with the past channel state, updates the current channel state. In the end, the decoder operates noncausally, observing the entire channel output sequence to produce its action sequence.

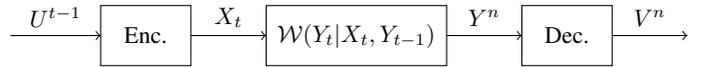
\begin{figure}[t]
  \centering
    \input{JSCC}
\caption{Joint source-channel coding with strictly causal encoder and noncausal decoder over Markov channel}
\vspace{-0.2cm}
\label{fig: JSCC Markov channel}
\end{figure} 

Empirical coordination coding framework, introduced in \cite{cuff2011coordination,raginsky2012empirical,cuff2010coordination}, characterizes cooperative behavior among agents by requiring the empirical distribution of symbol sequences to converge to a prescribed target distribution. In this framework, encoders and decoders select their actions based on local observations not only to ensure reliable communication, but also to induce a desired joint behavior. Coordination results for JSCC with (strictly) causal decision-makers over \ac{dmc} were established in \cite{cuff2011hybrid, Letreust2015empirical}. Extensions to point-to-point settings with channel state or feedback were developed in \cite{Treust2017joint,le2015empirical}, and recently have been further applied to distributed decision-making \cite{Treust2024power,zhao2024coordination,zhao2024causal,zhao2025zero}.

To the best of our knowledge, this work is the first to study JSCC over Markov channels through the lens of coordination coding. Under our assumption of the Markov channel model, we derive single-letter inner and outer bounds that characterize the set of achievable joint distributions that can be arbitrarily well approximated by suitable coding schemes. To establish the inner bound, we introduce a new notion of typicality, the input-driven Markov typicality, and develop its fundamental properties, such as the asymptotic equipartition property (AEP) and packing lemma. In contrast to the standard random-coding arguments for \ac{dmc} that rely on block independence in the block-Markov coding constructions, such as in \cite{le2015empirical, choudhuri2013causal}, this notion of typicality enables us to exploit the Markov channel structure directly, allowing the analysis to go beyond the independence-based arguments while still recovering the same results. The outer bound shares the same information constraint expression as the inner bound, but is defined over a generally larger set of distributions.




This paper is structured as follows: Our problem is formulated in Section \ref{sec: sys model}. Section \ref{sec: main results} presents our main results, which comprise the inner and outer bounds for the achievable distribution region. In Section \ref{app: input-driven typicality}, we introduce the input-driven Markov typicality and then prove the main results in Section \ref{app: achievability} and \ref{app: converse}.

\vspace{0.4cm}

\section{System Model}\label{sec: sys model}
We study the problem depicted in Fig. \ref{fig: JSCC Markov channel}. Capital letters, like $U$, denote random variables, calligraphic fonts, like $\mathcal{U}$, denote alphabets, and lowercase letters, like $u\in\mathcal{U}$, denote the realizations of random variables. We assume all the alphabets considered here $\mathcal{U},\mathcal{X},\mathcal{Y},\mathcal{V}$ are finite. The $\ell_1$ distance between two probability distributions $\mathcal{Q}$ and $\mathcal{P}$ is denoted by $\lVert\mathcal{Q}- \mathcal{P}\lVert_1 = \sum_{u}\abs{\mathcal{Q}(u) - \mathcal{P}(u)}$. The probability is denoted by $\mathbb P(\cdot)$. 

We assume the source is drawn i.i.d. $\sim\mathcal{P}_U$, and the channel is stationary Markov with one step memory of the latent state, with conditional distribution $\mathcal{W}_{Y|X,Y'}$ . Moreover, the statistics of the source and channel are known by both agents.

\begin{definition}
    A code with strictly causal encoder and noncausal decoder is a tuple of (deterministic) functions $c = (\{ f^{(t)}_{X_t|U^{t-1}}\}_{t=1}^n, g_{V^n|Y^n})$ defined by

    \begin{equation}
        f^{(t)}_{X_t|U^{t-1}}: \mathcal{U}^{t-1} \longrightarrow \mathcal{X},\quad g_{V^n|Y^n}: \mathcal{Y}^n\longrightarrow \mathcal{V}^n,
    \end{equation}
    where $U_0\triangleq \varnothing$.
    We denote by $\mathcal{C}(n)$ the set of codes with strictly causal encoder and non-causal decoder.
\end{definition}
We assume the starting channel state $Y_0$ is arbitrary. Then, a code $c\in\mathcal{C}(n)$ induces a joint distribution on sequences $(U^n,X^n,Y^n,V^n)$ given by
\begin{align}
    \mathcal{P}_{U}^{\otimes n}\prod_{t=1}^n\sbrackets{ f^{(t)}_{X_t|U^{t-1}}\cdot\mathcal{W}_{Y_{t}|X_t,Y_{t-1}}}g_{V^n|Y^n},\label{eq: joint sequential distribution}
\end{align}
from which we have the following:
\begin{align}
    &U_t\indep (X_t,Y_t),\label{eq: U_t indep X_t}\\
    &X_t-\!\!\!\!\minuso\!\!\!\!-U^{t-1}-\!\!\!\!\minuso\!\!\!\!- Y^{t-1},\label{eq: open-loop markov chain}\\
    &Y_t -\!\!\!\!\minuso\!\!\!\!- (X_t,Y_{t-1}) -\!\!\!\!\minuso\!\!\!\!- (U^{t-1},X^{t-1},Y^{t-2})\label{eq: markov chain, markov channel},
\end{align}
where \eqref{eq: U_t indep X_t} comes from strictly causal encoding and the independent source, \eqref{eq: open-loop markov chain} is the strictly causal encoder with no channel feedback, and \eqref{eq: markov chain, markov channel} is the Markov channel.

 We denote the empirical distribution $Q^n\in\Delta(\mathcal{U}\times\mathcal{X}\times\mathcal{Y}\times\mathcal{Y}\times\mathcal{V})$ of sequences $(u^n,x^n,y^n,v^n)$ with starting state $y_0$
 by
\begin{align}
        &Q^n(u,x,y',y,v)\label{eq: JSCC empirical distribution}\\
        &= \frac{1}{n}\sum_{t=1}^n\mathbf{1}\{u_t=u,x_t=x,y_{t-1}=y',y_t=y,v_t=v\},\nonumber\\
        &\quad\forall (u,x,y',y,v)\in\mathcal{U}\times\mathcal{X}\times\mathcal{Y}\times\mathcal{Y}\times\mathcal{V}.\nonumber
    \end{align}
Note here, we take into account the adjacent pairs $(y_{t-1},y_t)$ in the Markov sequence $y^n$.

\begin{definition}\label{def: achievable region}
    We define a joint distribution $\mathcal{P}_{U,X,Y',Y,V}\in\Delta (\mathcal{U}\times\mathcal{X}\times\mathcal{Y}\times\mathcal{Y}\times\mathcal{V})$ to be achievable, if for all $\varepsilon\geq 0$, there exists $N_0\in\mathbb N$, such that for all $n\geq N_0$,  we can find a sequence of codes $c\in\mathcal{C}(n)$, such that the induced joint empirical distribution given by \eqref{eq: JSCC empirical distribution} satisfies
    \begin{align}
        \mathbb P(||Q^n - \mathcal{P}_{U,X,Y',Y,V}||_{1}\geq\varepsilon)\leq \varepsilon.
    \end{align}
\end{definition}
Furthermore, we impose the following assumption:
\begin{assumption}
    For every stationary channel input distribution $\mathcal{P}_X$ on $\mathcal{X}$,
the induced Markov chain $Y^n=\{Y_1,\ldots,Y_n\}$ admits a unique recurrent class
$\mathcal{R}_Y\subseteq\mathcal{Y}$ that is irreducible
and aperiodic.
\end{assumption}

The above assumption is standard, see \cite{norris1998markov,breuer2005introduction,wu2025actions}, \cite[Sec.~4.1.2]{huang2015linear}. This assumption guarantees that the long-run behavior of the Markov chain $Y^n$ does not depend on its starting state and is ergodic within $\mathcal{R}_Y$. Moreover, it ensures the following result, see also \cite{norris1998markov},
\begin{lemma}
    The Markov chain $Y^n$ admits a unique equilibrium distribution $\pi_Y \in \Delta(\mathcal{R}_Y)$,
    which is uniquely determined by the stationary input distribution $\mathcal{P}_X$ and the channel transition kernel $\mathcal{W}_{Y|X,Y'}$ by
\begin{align}
\pi_Y(y)
= \sum_{x,y'}\pi_{Y}(y')\mathcal{P}_X(x)\mathcal{W}_{Y|X,Y'}(y|x,y')\label{eq: target equilibrium}.
\end{align}
\end{lemma}

Next, we present our main result: the single-letter inner and outer bounds of the achievable joint distribution region, defined in Definition \ref{def: achievable region}.

\section{Main Results}\label{sec: main results}


Our main result provides necessary and sufficient conditions for a distribution that is achievable.

\begin{theorem}[Main Result]\label{thm: main}
    A joint probability distribution $\mathcal{P}_{U,X,Y',Y,V}$ is achievable, if there exists an auxiliary random variable $W$, such that it factorizes as
    \begin{align}
        \mathcal{P}_U\mathcal{P}_X\mathcal{P}_{W|U,X}\pi_{Y'}\mathcal{W}_{Y|X,Y'}\mathcal{P}_{V|Y,X,W}\label{eq: achievable target distr},
    \end{align}
    and satisfies the information constraint
    \begin{align}
        I(X;Y|Y') - I(U;W|X) \geq 0,\label{eq: target info constraint}
    \end{align}
    where $\pi_{Y'}$ is the unique equilibrium distribution \eqref{eq: target equilibrium}. Conversely, $\mathcal{P}_{U,X,Y',Y,V}$ is achievable only if there exists an auxiliary random variable $W$ satisfies \eqref{eq: target info constraint} and the joint distribution decomposes as
    \begin{align}
        \mathcal{P}_U\mathcal{P}_X\mathcal{P}_{Y'|X}\mathcal{W}_{Y|X,Y'}\mathcal{P}_{W|U,X,Y,Y'}\mathcal{P}_{V|Y,X,W}\label{eq: converse target distr}.
    \end{align}
\end{theorem}

The first part of Theorem \ref{thm: main} is an inner bound of the achievable distribution region, whereas the second part is the outer bound. Although the inner and outer bounds share the same information constraint \eqref{eq: target info constraint}, the outer bound is generally looser, as it is defined over a larger class of joint distributions.

When the Markov channel reduces to a \ac{dmc} (e.g. take $Y'=\varnothing$), our inner bound reduces to the coordination coding result for joint source–channel coding with strictly causal encoders established in \cite[Theorem~3]{cuff2011hybrid}\footnote{Given the channel structure $(U,W) -\!\!\!\!\minuso\!\!\!\!- X -\!\!\!\!\minuso\!\!\!\!- Y$, both results coincide.}.

The auxiliary random variable $W$ plays a role analogous to that in coordination coding for \ac{dmc} and channels with state \cite{cuff2011hybrid,choudhuri2013causal,le2015empirical}: it facilitates the encoder to communicate compressed source information to the decoder in order to create the desired coordination. 
\vspace{-0.1cm}

\begin{remark}[Source-Channel Separation]
In the case when the random variables of the source $(U,V)$ are independent of those of the channel $(X,Y',Y)$, i.e., 
\begin{align*}
    \mathcal{P}_{U,X,Y',Y,V} = \mathcal{P}_{U,V}\cdot\mathcal{P}_{X,Y',Y},
\end{align*}
the information constraint \eqref{eq: target info constraint} reduces to a form analogous to Shannon’s source–channel separation theorem \cite{shannon1948mathematical}. In particular,
\vspace{-0.4cm}
\begin{align*}
    0&\leq I(X;Y|Y') - I(U;W|X)\\
    &\stackrel{\text{(a)}}{\leq} I(X;Y|Y') - I(U;V|X)\\
    & = I(X;Y|Y') - I(U;V)
\end{align*}
where the inequality (a) comes from the Markov chain $V-\!\!\!\!\minuso\!\!\!\!-(X,Y,W) -\!\!\!\!\minuso\!\!\!\!- (U,Y')$ at the noncausal decoder and the data processing inequality. Inequality (a) becomes an equality if and only if we take $W = U$. The resulting expression consists the channel capacity expression for Markov channel derived in \cite{wu2025actions} with feedback, minus the source rate required at the decoder.

\end{remark}

\vspace{-0.1cm}
The rest of the paper focuses on the proof of the main theorem. To establish the inner bound, we introduce a new notion of typicality -- the input-driven Markov typicality -- for sequences $(X^n,Y^n)$ in Sec. \ref{app: input-driven typicality}. The proof of the inner bound is a random-coding achievability argument, given in the Sec. \ref{app: achievability}. Finally, the outer bound is proved in Sec. \ref{app: converse}.

\vspace{0.3cm}


\section{Input-driven Markov Typicality}\label{app: input-driven typicality}
In this section, we introduce the input-driven Markov typicality for two sequences $(X^n,Y^n)$, as an extension to the strong Markov typicality for a single Markov sequence analyzed in \cite{csiszar2002method,davisson2003error,huang2015linear}.  We consider $X^n\sim\mathcal{P}_X^{\otimes n}$ i.i.d. generated, and $Y^n$, a Markov chain with an arbitrary starting state $Y_0$, and  $Y_t\sim\mathcal{W}_{Y|X,Y'}(\cdot|X_t,Y_{t-1})$ conditioning on $X_t$ for every $t=1,...,n$. In other words,
\vspace{-0.2cm}
\begin{align}\label{eq: x^ny^n generation}
    (X^n,Y^n)\sim\mathcal{P}_X^{\otimes n}\cdot\prod_{t=1}^n \mathcal{W}(Y_t|X_t,Y_{t-1}).
\end{align} 
\vspace{-0.2cm}

Now, for simplicity, we denote the target product distribution by $\mathcal{Q}_{Y'XY}=\pi_{Y'}\mathcal{P}_X\mathcal{W}_{Y|X,Y'}$. 

\begin{definition}[Input-driven Markov Typicality]
    For sequences $(x^n,y^n)$ with starting state $y_0$, define the empirical distribution of joint symbols over adjacent indices of triplets:
    \vspace{-0.1cm}
    \begin{align*}
        Q^n_{Y'XY}(i,x,j) =& \frac{1}{n}\sum_{t=1}^n\mathbf{1}\{(y_{t-1},x_t,y_t) = (i,x,j)\},\\
        & x\in\mathcal{X}, (i,j)\in\mathcal{Y}^2.
    \end{align*}
    Define the set of input-driven Markov typical sequences by
    \begin{align*}
            \mathcal{T}_{\varepsilon}^{(n)}(\mathcal{Q}_{Y'XY}) = \{(x^n,y^n): \lVert Q_{Y'XY}^n - \mathcal{Q}_{Y'XY}\lVert_1\leq \varepsilon\}.
    \end{align*} 
    For a fixed $x^n\in\mathcal{X}^n$, the set of conditionally Markov typical sequences $\mathcal{T}_\varepsilon^{(n)}(\mathcal{Q}_{Y'XY}|x^n)$ is defined by
    \begin{align*}
         &\mathcal{T}_{\varepsilon}^{(n)}(\mathcal{Q}_{Y'XY}|x^n) \\
         &= \{y^n\in\mathcal{Y}^n: \lVert Q_{Y'XY}^n - \pi_{Y'} Q_X^n\mathcal{W}_{Y|X,Y'}\lVert_1\leq \varepsilon\}.
    \end{align*}
    where \begin{align*}
        Q_X^n(x) = \frac{1}{n}\sum_{t=1}^n\mathbf{1}\{x_t=x\}= \sum_{i,j\in\mathcal{Y}}Q^n_{Y'XY}(i,x,j),\forall x\in\mathcal{X} .
    \end{align*}
\end{definition}

We have the following theorems for the input-driven Markov typicality:
\begin{theorem}[Ergodicity]\label{thm: conditional controlled Markov typicality lemma} Let $(X^n,Y^n)$ generated according to \eqref{eq: x^ny^n generation}, then
\begin{align}
    \forall\varepsilon>0, \lim_{n\rightarrow\infty}\mathbb P((X^n,Y^n)\in\mathcal{T}_\varepsilon^{(n)}(\mathcal{Q}_{Y'XY})) = 1.
\end{align}
In fact, it converges almost surely.
\end{theorem}
This theorem shows that no matter what the starting symbol is, the empirical distribution of $(X^n,Y^n)$ always converges to their joint stationary equilibrium distribution $\mathcal{Q}_{Y'XY}$.

We denote $H(X)$ by the entropy for random variable $X\sim\mathcal{P}_X$ and $H(Y|X,Y')$ for jointly discrete random variables $(Y',X,Y)\sim\mathcal{Q}_{Y'XY}$ which quantifies the remaining randomness in $Y$ given the independent input $X$ and the past symbol $Y'$ in the Markov chain. Now, we have

\begin{theorem}[AEP and Cardinality Bound]\label{lemma: AEP and cardinality bound}
    For any $\delta>0$, there exists $\varepsilon_0>0,N_0\in\mathbb N$, such that $\forall \varepsilon<\varepsilon_0$, $n>N_0$, and $\forall (x^n,y^n)\in\mathcal{T}_\varepsilon^{(n)}(\mathcal{Q}_{Y'XY})$, if a random sequence $(X^n,Y^n)$ is generated according to \eqref{eq: x^ny^n generation}, then it satisfies that
   \begin{enumerate}
  \item AEP: $
  \begin{aligned}[t]
    2^{-n(H(X,Y\mid Y')+\delta)}
      &< \mathbb{P}\bigl((X^n,Y^n)=(x^n,y^n)\bigr) \\
      &< 2^{-n(H(X,Y\mid Y')-\delta)},
  \end{aligned}
  $
  \item Cardinality bound: $\abs{\mathcal{T}_\varepsilon^{(n)}(\mathcal{Q}_{Y'XY})}<2^{n(H(X,Y|Y')+\delta)}$.
\end{enumerate}
\end{theorem}

More properties of the input-driven typicality and the proofs for the above theorems are given in the supplementary material available in \cite{zhao2026empirical}.

\vspace{0.5cm}
\section{Achievability Proof}\label{app: achievability}
The coding argument relies on the block-Markov coding scheme \cite{choudhuri2013causal}, extended to the coordination framework \cite{le2015empirical}. 

We consider a probability distribution $\mathcal{P}=\mathcal{P}_{U,X,Y',Y,V}$ decomposed as \eqref{eq: achievable target distr} that satisfies \eqref{eq: target info constraint} and \eqref{eq: target equilibrium}. There exists a $\delta>0$ and rate $R>0$ such that
    \begin{align}
        &R\geq I(U;W|X) + \delta,\label{eq: it constraint for covering lemma}\\
        &R\leq I(X;Y|Y') - \delta.\label{eq: it constraint for refined packing lemma}
    \end{align}
We consider a block-Markov random code $c\in\mathcal{C}(nB)$ over $B\in\mathbb N$ blocks of length $n\in\mathbb N$.

\textit{Random codebook:} We generate $|\mathcal{M}| = 2^{nR}$ sequences 
$X^n(m)$, where each drawn from the i.i.d. marginal distribution

\noindent$\mathcal{P}_X^{\otimes n}$ with index $m\in\mathcal{M}$. For each $m\in\mathcal{M}$, we generate the same number $|\mathcal{M}| = 2^{nR}$ of sequences $W^n(m,\hat{m})$ with index $\hat{m}\in\mathcal{M}$, drawn from the conditional distribution $\mathcal{P}_{W|X}^{\otimes n}$ conditioning on the sequence $X^n(m)$.

\textit{Encoding function:} Let $m_b$ denote the message generated during block $b\in[1:B]$. During the first block, without loss of generality, the encoder takes $m_1=1$ and returns $X^n(m_1)$. At the beginning of block $b\in\{2,...,B\}$, the encoder observes the sequence of source $U_{b-1}^n$ of the previous block $b-1$. It recalls the index $m_{b-1}\in\mathcal{M}$ of the sequence $X^n(m_{b-1})$ used for block $b-1$. It finds an index $m_b$ such that sequences
\begin{align*}
    (U^n_{b-1}, X^n(m_{b-1}),  W^n(m_{b-1}, m_b))\in\mathcal{T}_\varepsilon^{(n)}(\mathcal{P}_{U,X,W})
\end{align*}
are jointly typical. The encoder sends the sequence $X^n(m_b)$ corresponding to the current block $b$. We denote by $W_{b-1}^n\triangleq W^n(m_{b-1}, m_b),  X_b^n\triangleq X^n(m_b)$.

\textit{Channel Transmission: }During each block $b\in\{1,...,B\}$, at time instant $t=1$, the channel initializes by generating $Y_{b,1}\sim\mathcal{W}(\cdot|X_{b,1},Y_{b-1,n})$ conditioning on the current input $X_{b,1}$ and the last channel output $Y_{b-1,n}$ from the last block $b-1$. Then, the channel generates the current sequence $Y_{b,2}^n\sim\prod_{t=2}^n\mathcal{W}(Y_{b,t}|X_{b,t}, Y_{b,t-1})$ depending on $X_{b,2}^n$ corresponding to block $b$. Note that in this scheme, because the channel is Markov, the transmission for each block $b$ is \textit{not} independent.

\textit{Decoding function:} The decoder first returns $\Tilde{m}_1 = 1$. During block $b\in[2:B]$, the decoder recalls the past sequence $Y_{1,b-1}^n$ and the index $\Tilde{m}_{b-1}$ that corresponds to the sequence $\Tilde{X}_{b-1}^n = X^n(\Tilde{m}_{b-1})$. It observes the channel output $Y_{1,b}^n$ and finds the unique index $\Tilde{m}_b$ such that 
\begin{align*}
  &(Y_{b}^n, X^n(\Tilde{m}_b))\in\mathcal{T}_\varepsilon^{(n)}(\mathcal{P}_{X,Y',Y}),\\
  &(Y_{b-1}^n, X^n(\Tilde{m}_{b-1}), W^n(\Tilde{m}_{b-1}, \Tilde{m}_b))\in\mathcal{T}_\varepsilon^{(n)}(\mathcal{P}_{X,W,Y',Y})
\end{align*}
are input-driven Markov typical. We denote by $\Tilde{X}_{b}^n = X^n(\Tilde{m}_{b})$ and  $\Tilde{W}_{b-1}^n = W^n(\Tilde{m}_{b-1}, \Tilde{m}_{b})$ as our choice. The decoder non-causally generates $V_{b}^n\sim \mathcal{P}_{V|Y,X,W}^{\otimes n}$ depending on sequences $(Y_{b}^n, \Tilde{X}_{b}^n, \Tilde{W}_{b}^n)$ for $b\in[1:B-1]$. As for the last block, the decoder simply outputs an all zero sequence, i.e.,  $V_{B}^n = \mathbf{0}$. Usually, sequences are \textit{not} jointly typical in the last block.

\textit{Error Analysis: }Next, we show that given the above coding scheme, the generated sequences $(U_{b}^n, X_b^n, W_b^n, Y_b^n, V_b^n)$ for each block $b\in\{1,...,B-1\}$ are jointly (Markov) typical with respect to the joint distribution given in \eqref{eq: achievable target distr} with high probability. For every $\varepsilon >0$, there exists $N_0\in\mathbb N$, such that the probability of error events for all $n\geq N_0$:
\begin{align*}
    &\text{(a) } \mathbb E_c\left[\mathbb P\left(\forall m\in\mathcal{M}, (U^n_{b-1}, X^n(m_{b-1}),  W^n(m_{b-1}, m))\right.\right.\\
    &\qquad\qquad\qquad\qquad\qquad\qquad\qquad\notin\mathcal{T}_\varepsilon^{(n)}(\mathcal{P}_{UXW}))]\leq\frac{\varepsilon}{3},\\
    &\text{(b) } \mathbb E_c\sbrackets{\mathbb P\brackets{(X_b^n,  Y_b^n)\notin\mathcal{T}_\varepsilon^{(n)}(\mathcal{P}_{Y'XY})}}\leq \frac{\varepsilon}{3},\\
    &\text{(c) } \mathbb E_c\left\{\mathbb P\left(\exists  \Tilde{m}\neq m, \text{s.t. }\{(Y_{b}^n, X^n(\Tilde{m}))\in\mathcal{T}_\varepsilon^{(n)}(\mathcal{P}_{Y'XY})\} \cap \right.\right. \\
    &\left.\left.\!\!\!\!\!
    \{(Y_{b-1}^n, X^n(m_{b-1}), W^n(m_{b-1}, \Tilde{m}))\!\in\!\mathcal{T}_\varepsilon^{(n)}(\mathcal{P}_{WY'XY})\} \right)\right\}\!\leq\! \frac{\varepsilon}{3}.
\end{align*}
Here, (a) is guaranteed by covering lemma for classical i.i.d. sequences \cite{cover1999elements} and the condition \eqref{eq: it constraint for covering lemma}. Item (b) is guaranteed by the sequential generation of $(X^n,Y^n)$ and Theorem \ref{thm: conditional controlled Markov typicality lemma} for the input-driven Markov typicality. Item (c) is satisfied by a joint packing lemma shown in Sec. \ref{app sub: joint packing lemma} below, such that both (dependent) events hold when the rate constraint \eqref{eq: it constraint for refined packing lemma} is met.

We denote by $\Tilde{Q}^n\in\Delta(\mathcal{U}\times\mathcal{X}\times\mathcal{W}\times\mathcal{R}_Y\times\mathcal{R}_Y\times\mathcal{V})$ the empirical distribution of symbols $(U,X,W,Y',Y,V)$ over blocks $b=1,...,B-1$, where $(Y',Y)$ are the adjacent symbols of the Markov chain. We now show that $\Tilde{Q}^n$ is close to the empirical distribution $Q^n$ over all $B$ blocks for $B$ sufficiently large. We denote by $Q_B$ the empirical distribution of all symbols over the last block. Now,
\begin{align*}
    \lVert Q^n - \Tilde{Q}^n\lVert_{1}
    & = \lVert \frac{1}{B}\brackets{(B-1)\Tilde{Q}^n + Q_B} - \Tilde{Q}^n\lVert_{1}\\
    & = \frac{1}{B}\lVert Q_B - \Tilde{Q}^n\lVert_{1} \\
    &\leq \frac{2}{B}\cdot\abs{\mathcal{U}\times\mathcal{X}\times\mathcal{W}\times\mathcal{R}_Y\times\mathcal{R}_Y\times\mathcal{V}} \leq \varepsilon
\end{align*}
when $B\geq  \frac{2}{\varepsilon}\abs{\mathcal{U}\times\mathcal{X}\times\mathcal{W}\times\mathcal{R}_Y\times\mathcal{R}_Y\times\mathcal{V}}$. Then, the expected error probability
\begin{align*}
    \mathbb E_c[\mathbb P_e(c)]&= \mathbb E_c\sbrackets{\mathbb P\brackets{\lVert Q^n-\mathcal{P}\lVert_{1}\geq 2\varepsilon}}\\
    &\leq\mathbb E_c\sbrackets{\mathbb P\brackets{\lVert Q^n-\Tilde{Q}^n\lVert_{1}+\lVert\Tilde{Q}^n-\mathcal{P}\lVert_{1}\geq 2\varepsilon}}\\
    &\leq \mathbb E_c\sbrackets{\mathbb P\brackets{\lVert \Tilde{Q}^n-\mathcal{P}\lVert_{1}\geq \varepsilon}}\leq \varepsilon.
\end{align*}
This implies the existence of a code $c^*\in\mathcal{C}(nB)$ with an error probability below $\varepsilon$ for all $n\geq N_0B$.\hfill \qedsymbol{}


\subsection{Joint Packing Lemma}\label{app sub: joint packing lemma}
We denote the following events for block $b$ by
\begin{align*}
    &A_b(\Tilde{m})\triangleq \{(Y_{b}^n, X^n(\Tilde{m}))\in\mathcal{T}_\varepsilon^{(n)}(\mathcal{P}_{Y'XY})\},\\
    &B_{b}(\Tilde{m})\triangleq\{(Y_{b}^n, X^n(m_{b}), W^n(m_{b}, \Tilde{m}))\in\mathcal{T}_\varepsilon^{(n)}(\mathcal{P}_{WY'XY})\}.
\end{align*}
Note that, compared to the joint packing lemma for the \ac{dmc} scenario, e.g., in \cite{le2015empirical}, here, since the channel has memory, $A_b(\Tilde{m})$ and $B_{b-1}(\Tilde{m})$ are not independent anymore, due to the channel state $Y_{b-1,n}$. However, by conditioning on $Y_{b-1,n}=y'$, applying the Markov property, we can have conditional independence and therefore show a similar result.

\begin{lemma}[Joint Packing Lemma]
For each block $b\in [2:B]$ and for each previous message $m_{b-1}$, if $R\leq  I(X;Y|Y') - \delta$, then, $\forall \varepsilon>0$, $\exists N_0\in\mathbb N$, such that for all $n\geq N$, 
\begin{align*}
   \mathbb E_c\left\{\mathbb P\left(\exists  \Tilde{m}\neq m, \text{s.t. }\ A_b(\Tilde{m}) \cap B_{b-1}(\Tilde{m})\right)\right\}\leq \varepsilon
\end{align*}
\end{lemma}
\noindent\textit{Proof. }Consider $\varepsilon>0$ that satisfies $4\varepsilon <\delta$. We have
\begin{align*}
    R - I(Y;X|Y')  + 3\varepsilon \leq -\delta + 3\varepsilon <-\varepsilon.
\end{align*}
    \begin{align*}
    &\mathbb E_c\left\{\mathbb P\left(\exists  \Tilde{m}\neq m, \text{s.t. }\ A_b(\Tilde{m}) \cap B_{b-1}(\Tilde{m})\right)\right\} \\
    &\stackrel{}{\leq} \sum_{\Tilde{m}\neq m}\mathbb E_c\left\{\mathbb P\left(\ A_b(\Tilde{m}) \cap B_{b-1}(\Tilde{m})\right)\right\}\\
    &\stackrel{}{=}\sum_{\Tilde{m}\neq m}\mathbb E_c\left\{\sum_{y}\mathbb P(Y_{b-1,n} = y)\right.\\
    &\\
    &\\
    &\qquad\qquad\left.
    \cdot\mathbb P\left( A_b(\Tilde{m}) \cap B_{b-1}(\Tilde{m})\Bigg |Y_{b-1,n} = y \right)\right\}\\
    &= \sum_{\Tilde{m}\neq m}\mathbb E_c\left\{\sum_{y}\mathbb P(Y_{b-1,n} = y)\cdot\mathbb P\left( B_{b-1}(\Tilde{m})\Bigg | Y_{b-1,n} = y \right)\right.\\
    &\qquad\qquad\left.\cdot\mathbb P\left(A_b(\Tilde{m})\Bigg | B_{b-1}(\Tilde{m}), Y_{b-1,n} = y  \right)\right\}\\
    &\stackrel{\text{(a)}}{=}\sum_{\Tilde{m}\neq m}\mathbb E_c\left\{\sum_{y}\mathbb P(Y_{b-1,n} = y)\cdot\mathbb P\left( B_{b-1}(\Tilde{m})\Bigg |Y_{b-1,n} = y \right)\right.\\
    &\qquad\qquad\left. \cdot\mathbb P\left(A_b(\Tilde{m})\Bigg|  Y_{b-1,n} = y  \right)\right\}\\
    &=\sum_{\Tilde{m}\neq m}\mathbb E_c\left\{\sum_{y}\mathbb P(Y_{b-1,n} = y)\cdot\mathbb P\left( B_{b-1}(\Tilde{m})\Bigg |Y_{b-1,n} = y \right)\right.\\
     &\quad \left.\!\!\cdot \!\!\sum_{(y^n,x^n)\in\mathcal{T}_\varepsilon^{(n)}(\mathcal{P})}\mathbb P\left(Y_{b}^n=y^n X^n(\Tilde{m}))=x^n\Bigg|  Y_{b-1,n} = y  \right)\right\}\\
     &\stackrel{\text{(b)}}{=}\sum_{\Tilde{m}\neq m}\mathbb E_c\left\{\sum_{y}\mathbb P(Y_{b-1,n} = y)\cdot\mathbb P\left( B_{b-1}(\Tilde{m})\Bigg |Y_{b-1,n} = y \right)\right.\\
     &\quad\left.\cdot \!\!\!\!\!\!\!\!\!\!\sum_{(y^n,x^n)\in\mathcal{T}_\varepsilon^{(n)}(\mathcal{P})}\mathbb P\!\!\left(\!\!Y_{b}^n=y^n \Bigg|  Y_{b-1,n}\! =\! y\!  \right)\!\!\cdot\! \mathbb P\left(X^n(\Tilde{m}))=x^n\right)\right\}\\
      &\stackrel{\text{(c)}}{\leq }\sum_{\Tilde{m}\neq m}\mathbb E_c\left\{\sum_{y}\mathbb P(Y_{b-1,n} = y)\cdot\mathbb P\left( B_{b-1}(\Tilde{m})\Bigg |Y_{b-1,n} = y \right)\right.\\
      &\quad\quad \left.\cdot \sum_{(y^n,x^n)\in\mathcal{T}_\varepsilon^{(n)}(\mathcal{P})} 2^{-n(H(Y|Y') - \varepsilon)}\times 2^{-n(H(X) - \varepsilon)}  \right\}\\
      &\stackrel{\text{(d)}}{\leq } 2^{n(H(Y,X|Y') - H(Y|Y') - H(X) + 3\varepsilon)}\cdot \sum_{\Tilde{m}\neq m}\mathbb E_c\left\{\mathbb P( B_{b-1}(\Tilde{m}) ) \right\}\\
      &\stackrel{\text{(e)}}{\leq }2^{nR}\times 2^{n(-I(X;Y|Y')+3\varepsilon)}\\
                &\leq 2^{-n\varepsilon},
\end{align*}
where,\\
 (a) Markov property of the channel output generation: the dependence of the current block's sequence on the last block only through its last symbol $Y_{b-1,n}$;\\
 (b) the independence of the random variable $Y_{b}^n$ given $Y_{b-1,n}$ with $X^n(\Tilde{m})$ where $\Tilde{m}\neq m$;\\
 (c) AEP of the strong Markov typicality for Markov sequence $Y^n$, which leverages ergodicity properties that the effect of the boundary initial state can be averaged out, see \cite[Proposition 4.1.1]{huang2015linear}, and also $X\indep Y'$;\\
 (d) cardinality bound for the input-driven Markov typical sequence, i.e., Theorem \ref{lemma: AEP and cardinality bound};\\
 (e) $\mathbb E_c\left\{\mathbb P( B_{b-1}(\Tilde{m}) ) \right\}\leq 1$ and the number of codewords $|\mathcal{M}|=2^{nR}$.\hfill \qedsymbol{}

\vspace{0.9cm}

\section{Converse Proof}\label{app: converse}
For $n\in\mathbb N$, we consider a code $c\in\mathcal{C}(n)$ that induces a joint sequential distribution of form \eqref{eq: joint sequential distribution}. Then,
\begin{align*}
    &0
    = I(U^n;Y^n) - I(U^n;Y^n)\\
    &\stackrel{\text{(a)}}{=} \sum_t I(U^n;Y_t|Y^{t-1}) - \sum_t I(U_t; Y^n|U^{t-1})\\
    &\stackrel{\text{(b)}}{=} \sum_t I(U^{t-1};Y_t|Y^{t-1}) - \sum_t I(U_t; Y^n|U^{t-1})\\
    &\leq \sum_t I(U^{t-1};Y_t|Y^{t-1}) - \sum_t I(U_t; Y^{t-1}, Y^n_{t+1}|U^{t-1})\\
    &\stackrel{\text{(c)}}{=} \sum_t I(X_t, U^{t-1};Y_t|Y^{t-1})\! -\! \sum_t I(U_t; \!Y^{t-1},\! Y^n_{t+1}|U^{t-1},X_t)\\
    &\stackrel{\text{(d)}}{=} \sum_t I(X_t, U^{t-1};Y_t|Y^{t-1}) \!-\! \sum_t I(U_t;U^{t-1},\! Y^{t-1}, \!Y^n_{t+1}|X_t)\\
    & \stackrel{\text{(e)}}{\leq}  \sum_t I(X_t;Y_t|Y_{t-1}) - \sum_t I(U_t;W_t|X_t)\\
     & \stackrel{\text{(f)}}{=} n\cdot\brackets{I(X_T;Y_T|Y_{T-1},T) - I(U_T;W_T|X_T,T)}\\
     & \stackrel{}{=} n\cdot(H(Y_T|Y_{T-1},T) - H(Y_T|X_T,Y_{T-1},T) \\
     &\qquad\qquad\qquad\qquad\qquad  -I(U_T;W_T,T|X_T))\\
     & \stackrel{\text{(g)}}{\leq} n\cdot(H(Y_T|Y_{T-1})\! -\! H(Y_T|X_T,\!Y_{T-1})\! -\!I(U_T;W_T,\!T|X_T))\\
      & =n\cdot\brackets{I(X_T;Y_T|Y_{T-1}) - I(U_T;W_T,T|X_T)}\\
      & \stackrel{\text{(h)}}{=} n\cdot\brackets{I(X;Y|Y') - I(U;W|X)},
\end{align*}
where,\\
(a) chain rule of mutual information;\\
(b) because the source is generated i.i.d., we have $U_t^n\indep (U^{t-1},Y^{t-1},Y_t)$, hence $I(U_t^n;Y_t|Y^{t-1},U^{t-1}) = 0$;\\
(c) deterministic strictly causal encoding function, i.e., $X_t = f^{(t)}(U^{t-1})$;\\
(d) i.i.d. source and strictly causal encoding, we hence have $I(U_t;U^{t-1}|X_t) = 0$;\\
(e) because of conditioning reduces entropy, and the Markov chain $Y_t -\!\!\!\!\minuso\!\!\!\!- (X_t,Y_{t-1}) -\!\!\!\!\minuso\!\!\!\!- (U^{t-1},Y^{t-2})$. Also, in the second term, we identify  $W_t = (U^{t-1},Y^{t-1}, Y_{t+1}^n)$;\\
(f) the introduction of the uniform random variable $T$ over $\{1,...,n\}$ and the introduction of the corresponding mean random variables $U_T,X_T,W_T,Y_{T},V_T$ and $Y_{T-1}$ represents the previous symbol of the current state $Y_T$;\\
(g) conditioning reduces entropy, and stationarity of the Markov channel;\\
(h) the assignment of $U=U_T, X = X_T, W=(W_T,T), Y'=Y_{T-1}, Y=Y_T,V=V_T$.

Lastly, we show that, random variables $(U,X,W,Y',Y,V)$ defined above satisfy the following Markov chains
\begin{itemize}
    \item $U\indep(X,Y',Y)$ comes from the strictly causal encoding,
    \item $V -\!\!\!\!\minuso\!\!\!\!- (X,Y,W) -\!\!\!\!\minuso\!\!\!\!- (U,Y')$ comes from the noncausal decoding, and the fact that $Y^n$ is included in $(W_t,Y_t)$ for all $t\in\{1,...,n\}$, hence is included in $(W_T,T,Y_T) = (W,Y)$.\hfill \qedsymbol{}
\end{itemize}

\bibliographystyle{ieeetr}
\bibliography{IEEEabrv,main}

\input{appendix}
\end{document}

%% file: JSCC.tex


\centering
\begin{tikzpicture}[scale=0.8,every node/.style={scale=0.9}]

    \draw (0,0) rectangle (1.5,1);
    \draw (2.75,0) rectangle (5.75,1);

    \draw (7,0) rectangle (8.5,1);

    \draw[->] (-1.5,0.5) -- (0,0.5);
    \draw[->] (1.5,0.5) -- (2.75,0.5);
    \draw[->] (5.75,0.5) -- (7,0.5);
    \draw[->] (8.5,0.5) -- (9.75,0.5);

    \node at (0.75,0.5) {Enc.};

    \node at (2.125,0.8) {$X_t$};

    \node at (4.25,0.5) {$\mathcal{W}(Y_{t}|X_t,Y_{t-1})$};
    
    \node at (6.375,0.8) {$Y^n$};
   
    \node at (7.75,0.5) {Dec.};
    \node at (9.125,0.8) {$V^n$};

    \node at (-0.75,0.8) {$U^{t-1}$};
    
\end{tikzpicture}



%% file: appendix.tex
\newpage
\onecolumn

\appendices
\section{Property of the Input-Driven Markov Typicality}

Under Assumption A, the induced transition matrix $T = [T_{i,j}]_{i,j\in\mathcal{Y}}$ of the Markov chain $Y^n$ does not depend on time $t\in[1:n]$, where
\begin{align*}
    T_{i,j} &= \mathbb P(Y_k=j|Y_{k-1}=i)\\
    &= \sum_{x\in\mathcal{X}}\mathcal{P}_{X}(x)\mathcal{W}(Y_k=j|X_k=x,Y_{k-1}=i).\label{eq: transition distribution of markov Y}
\end{align*}

We can prove the following property of the input-driven Markov typicality:

\begin{proposition}\label{prop: joint controlled Markov typical}
    For $n\in\mathbb N$ and $\varepsilon>0$, the typical sequences satisfy the following properties\\
    1) If $(x^n,y^n)\in\mathcal{T}_\varepsilon^{(n)}(\mathcal{Q}_{Y'XY})$, then $x^n\in\mathcal{T}_\varepsilon^{(n)}(\mathcal{ P}_X)$, $y^n\in\mathcal{T}_\varepsilon^{(n)}(\pi\cdot T)$, $y^n\in\mathcal{T}_{2\varepsilon}^{(n)}(\mathcal{Q}_{Y'XY}|x^n)$ conditional typical, where $T$ is the marginalized transition matrix of Markov sequence $Y^n$ given in \eqref{eq: transition distribution of markov Y}.\\
    2) If $x^n\in\mathcal{T}_\varepsilon^{(n)}(\mathcal{ P}_X)$, $y^n\in\mathcal{T}_{2\varepsilon}^{(n)}(\mathcal{Q}_{Y'XY}|x^n)$, then $(x^n,y^n)\in\mathcal{T}_{2\varepsilon}^{(n)}(\mathcal{Q}_{Y'XY})$
\end{proposition}

\begin{proof}
    1) Proof of item 1: If $(x^n,y^n)\in\mathcal{T}_\varepsilon^{(n)}(\mathcal{Q}_{Y'XY})$, then
    \begin{align*}
        \varepsilon&\geq\lVert Q_{Y'XY}^n - \mathcal{Q}_{Y'XY}\lVert_1\\
        &=\sum_{x\in\mathcal{X},i,j\in\mathcal{Y}}\abs{Q_{Y'XY}^n(i,x,j) - \pi_{Y'}(i)\mathcal{P}_X(x)\mathcal{W}(j|x,i)}\\
        &\geq \sum_x\abs{\sum_{i,j}\brackets{Q_{Y'XY}^n(i,x,j) - \pi_{Y'}(i)\mathcal{P}_X(x)\mathcal{W}(j|x,i)}}\\
        & = \sum_x\abs{Q_X^n(x) - \mathcal{P}_X(x)}\\
        &=\lVert Q_{X}^n - \mathcal{P}_X\lVert_1
    \end{align*}
    Similarly,
    \begin{align*}
        \varepsilon&\geq\lVert Q_{Y'XY}^n - \mathcal{Q}_{Y'XY}\lVert_1\\
        &=\sum_{x\in\mathcal{X},i,j\in\mathcal{Y}}\abs{Q_{Y'XY}^n(i,x,j) - \pi_{Y'}(i)\mathcal{P}_X(x)\mathcal{W}(j|x,i)}\\
        &\geq \sum_{i,j}\abs{\sum_{x}\brackets{Q_{Y'XY}^n(i,x,j) - \pi_{Y'}(i)\mathcal{P}_X(x)\mathcal{W}(j|x,i)}}\\
        & = \sum_{i,j}\abs{Q_Y^n(i,j) - \pi_{Y'}(i) T(j|i)}\\
        &= \lVert Q_{Y}^n - \pi\cdot T\lVert_1
    \end{align*}
    Moreover, we also have
    \begin{align*}
        &\lVert Q_{Y'XY}^n - \pi_{Y'} Q_X^n\mathcal{W}\lVert_1\\
        &=\lVert Q_{Y'XY}^n- \mathcal{Q}_{Y'XY} +  \mathcal{Q}_{Y'XY} - \pi_{Y'} Q_X^n\mathcal{W}\lVert_1\\
        &\leq \lVert Q_{Y'XY}^n- \mathcal{Q}_{Y'XY}\lVert_1 + \lVert\mathcal{Q}_{Y'XY} - \pi_{Y'} Q_X^n\mathcal{W}\lVert_1\\
        &\leq \lVert Q_{XY}^n- \mathcal{Q}_{Y'XY}\lVert_1 + \sum_{x,i,j}\abs{\pi_{Y'}(i)\mathcal{P}_X(x)\mathcal{W}(j|x,i) - \pi_{Y'}(i)Q_X^n(x)\mathcal{W}(j|x,i)}\\
        &\leq \lVert Q_{XY}^n- \mathcal{Q}_{Y'XY}\lVert_1 + \sum_{x,i,j}\abs{\mathcal{P}_X(x) - Q_X^n(x)}\cdot\brackets{\pi_{Y'}(i)\mathcal{W}(j|x,i)}\\
        &= \lVert Q_{XY}^n- \mathcal{Q}_{Y'XY}\lVert_1 + \sum_{x}\abs{\mathcal{P}_X(x) - Q_X^n(x)}\cdot \sum_{i,j}\brackets{\pi_{Y'}(i)\mathcal{W}(j|x,i)}\\
        &=\lVert Q_{XY}^n- \mathcal{Q}_{Y'XY}\lVert_1 + \lVert \mathcal{P}_X - Q_X^n\lVert_1\\
        &\leq 2\cdot \varepsilon.
    \end{align*}

    2) Proof of item 2: 
    \begin{align*}
        &\lVert Q_{Y'XY}^n - \pi_{Y'}\mathcal{P}_X\mathcal{W}\lVert_1\\
        &=\lVert Q_{Y'XY}^n - \pi_{Y'} Q_X^n\mathcal{W} + \pi_{Y'} Q_X^n\mathcal{W} -\pi_{Y'}\mathcal{P}_X\mathcal{W}\lVert_1\\
        &\leq \lVert Q_{Y'XY}^n - \pi_{Y'} Q_X^n\mathcal{W}\lVert_1 + \lVert \pi_{Y'} Q_X^n\mathcal{W} -\pi_{Y'}\mathcal{P}_X\mathcal{W}\lVert_1\\
        &\leq 2\varepsilon.
    \end{align*}
\end{proof}

\section{Proof for Theorem \ref{thm: conditional controlled Markov typicality lemma}}

Define the stochastic process for $t\geq 1$:
    \begin{align}
        S_t = (Y_{t-1},X_t,Y_t)
    \end{align}
    It is obviously a Markov chain.

    We show first, that $\{S_t\}$ is a time-homogeneous Markov chain: Given the current state $S_t = (Y_{t-1} = i,X_t = x,Y_t = j)$, then for the next time
    \begin{itemize}
        \item $X_{t+1}\sim\mathcal{P}_X$ independent of the past,
        \item $Y_{t+1}\sim\mathcal{W}(\cdot|X_{t+1}, Y_t = j)$
    \end{itemize}
    Therefore, the transition probability
    \begin{align}
        \mathbb P(S_{t+1} = (j,x',k)|S_t = (i,x,j)) = \mathcal{P}_X(x')\mathcal{W}(k|x',j)
    \end{align}
    does not depend on time. Hence $S$ is homogeneous. Moreover, since $\{Y_k\}$ has only one irreducible aperiodic recurrent set $\mathcal{R}_Y$, so does $\{S_k\}$, and it is simply $\mathcal{R}_S = \mathcal{R}_Y\times\mathcal{X}\times \mathcal{R}_Y$. Therefore, $S$ has a unique equilibrium distribution on $\mathcal{R}_S$. Next, we prove the equilibrium distribution of the process $S$ is simply $\mathcal{Q}_{Y'XY}=\pi_{Y'}\mathcal{P}_X\mathcal{W}_{Y|X,Y'}$, i.e., we want to have
    \begin{align}
        \sum_{i,x} \mathcal{Q}_{Y'XY}(i,x,j)\underbrace{\mathcal{P}_X(x')\mathcal{W}(k|x',j)}_{\text{transition matrix of $S$}} = \mathcal{Q}_{Y'XY}(j,x',k)
    \end{align}
    Now, because
    \begin{align*}
        &\sum_{i,x} \mathcal{Q}_{Y'XY}(i,x,j)\mathcal{P}_X(x')\mathcal{W}(k|x',j)\\
        &=\brackets{\sum_{i,x}\pi(i)\mathcal{P}_X(x)\mathcal{W}(j|x,i)}\mathcal{P}_X(x')\mathcal{W}(k|x',j)\\
        &=\brackets{\sum_{i}\pi(i)\sum_x\sbrackets{\mathcal{P}_X(x)\mathcal{W}(j|x,i)}}\mathcal{P}_X(x')\mathcal{W}(k|x',j)\\
        &=\pi(j)\mathcal{P}_X(x')\mathcal{W}(k|x',j)\\
        &=\mathcal{Q}_{Y'XY}(j,x',k).
    \end{align*}
    $\mathcal{Q}_{Y'XY}=\pi_{Y'}\mathcal{P}_X\mathcal{W}_{Y|X,Y'}$ is the invariant distribution of $S$. Since $\{S_t\}$ is a Markov chain with finite states and a unique irreducible aperiodic recurrent class, it is then ergodic. By the ergodic theorem, for each state $(i,x,j)\in\mathcal{R}_S$, we have
    \begin{align}
        \frac{1}{n}\sum_{t=1}^n \mathbf{1}\{S_t = (i,x,j)\} \xrightarrow[]{a.s.}\mathcal{Q}_{Y'XY}(i,x,j)
    \end{align}
    where the left hand side above is exactly the joint empirical $Q_{Y'XY}^n(i,x,j)$. Therefore
    \begin{align}
        \lVert Q_{Y'XY}^n - \mathcal{Q}_{Y'XY}\lVert_1 \xrightarrow[]{a.s.} 0.
    \end{align}
    Since almost sure convergence implies converge in probability, we have 
    \begin{align}
    \forall\varepsilon>0, \lim_{n\rightarrow\infty}\mathbb P((X^n,Y^n)\in\mathcal{T}_\varepsilon^{(n)}(\mathcal{Q}_{Y'XY}) = 1.
\end{align}
\hfill \qedsymbol{}

\section{Proof of Theorem \ref{lemma: AEP and cardinality bound}}
\noindent \textit{Proof of 1)}

    Let $\mathbb P((X_1,Y_1) = (x_1,y_1))  =  c$. Then,
    \begin{align*}
        &\mathbb P((X^n,Y^n) = (x^n,y^n))\\
        &=\mathbb P((X_1,Y_1) = (x_1,y_1))\prod_{i=2}^n\mathcal{P}_X(x_i)\mathcal{W}(y_i|x_i,y_{i-1})\\
        &=c\prod_{x\in\mathcal{X}}\mathcal{P}_X(x)^{N(x|x^n)}\prod_{x\in\mathcal{X},i,j\in\mathcal{Y}}\mathcal{W}(j|x,i)^{N(i,x,j|x^n,y^n)}
    \end{align*}
    where $N(x|x^n) = \sum_{t=1}^n\mathbf{1}\{x_t=x\},N(i,x,j|x^n,y^n) = \sum_{t=1}^n\mathbf{1}\{(y_{t-1},x_t,y_t) = (i,x,j)$.
    
    Now, take $\log_2$ on both sides and multiply by $-\frac{1}{n}$. Then we have
    \begin{align}
        &-\frac{1}{n}\log_2\mathbb P((x^n,y^n))\nonumber\\
        &= -\frac{1}{n}\log_2c - \sum_x Q_X^n(x)\log_2\mathcal{P}_X(x) - \sum_{x,i,j}Q^n_{Y'XY}(i,x,j)\log_2\mathcal{W}(j|x,i)\nonumber\\
        &= -\frac{1}{n}\log_2c - \sum_x \brackets{Q_X^n(x)-\mathcal{P}_X(x)}\log_2\mathcal{P}_X(x) - \sum_x\mathcal{P}_X(x)\log_2\mathcal{P}_X(x)\label{eq: eq to continue}\\
        &\quad - \sum_{x,i,j}\brackets{Q^n_{Y'XY}(i,x,j) - \pi_{Y'}(i)\mathcal{P}_X(x)\mathcal{W}(j|x,i)}\log_2\mathcal{W}(j|x,i) - \sum_{x,i,j} \pi_{Y'}(i)\mathcal{P}_X(x)\mathcal{W}(j|x,i)\log_2\mathcal{W}(j|x,i)\nonumber
    \end{align}
    Because $\mathcal{X},\mathcal{Y}$ are finite, then
    \begin{align}
        L_X:= \max_{x:\mathcal{P}_X(x)>0}|\log_2\mathcal{P}_X(x)|,\quad L_W:= \max_{(i,x,j):\mathcal{W}(j|x,i)>0}|\log_2\mathcal{W}(j|x,i)|
    \end{align}
    are both finite numbers.
    Now, since $(x^n,y^n)\in\mathcal{T}_\varepsilon^{(n)}(\mathcal{Q}_{Y'XY})$ input-driven Markov typical, we also have $\lVert Q_X^n - \mathcal{P}_X\lVert_1\leq \varepsilon$ provided by Proposition \ref{prop: joint controlled Markov typical}. Therefore, we obtain that
    \begin{align*}
       -\varepsilon\cdot L_X \leq\sum_x \brackets{Q_X^n(x)-\mathcal{P}_X(x)}\log_2\mathcal{P}_X(x) \leq \varepsilon\cdot L_X 
    \end{align*}
    and
    \begin{align*}
        -\varepsilon\cdot L_W\leq\sum_{x,i,j}\brackets{Q^n_{Y'XY}(i,x,j) - \pi_{Y'}(i)\mathcal{P}_X(x)\mathcal{W}(j|x,i)}\log_2\mathcal{W}(j|x,i)\leq \varepsilon\cdot L_W.
    \end{align*}

    Now, for any fixed $\delta>0$, choose $N_0$ big enough such that 
    \begin{align*}
        -\delta/3\leq -\frac{1}{n}\log_2c\leq \delta/3,\quad \forall n>N_0,
    \end{align*}
    and $\varepsilon_0$ small enough such that
    \begin{align*}
        \varepsilon_0(L_X+L_W)\leq \delta/3.
    \end{align*}
    Then, $\forall \varepsilon<\varepsilon_0$ and $n>N_0$, from \eqref{eq: eq to continue} we obtain that every $(x^n,y^n)\in\mathcal{T}_\varepsilon^{(n)}(\mathcal{Q}_{Y'XY})$ satisfies
    \begin{align*}
        &-\frac{1}{n}\log_2\mathbb P((x^n,y^n))< \delta + H(X)+H(Y|X,Y')=\delta + H(X,Y|Y')\\
        &-\frac{1}{n}\log_2\mathbb P((x^n,y^n))> -\delta+ H(X)+H(Y|X,Y')=-\delta + H(X,Y|Y')
    \end{align*}
    Hence, for every $(x^n,y^n)\in\mathcal{T}_\varepsilon^{(n)}(\mathcal{Q}_{Y'XY})$,
    \begin{align*}
        2^{-n\sbrackets{H(X,Y|Y')+\delta}}< \mathbb P((x^n,y^n)) < 2^{-n\sbrackets{H(X,Y|Y')-\delta}}
    \end{align*}

    Item 2) follows from the lower bound in item 1).\hfill \qedsymbol{}